\documentclass{vldb}
\usepackage{graphicx,xspace,url}
\usepackage{balance}  
\usepackage{xcolor}
\usepackage{hyperref}
\usepackage{booktabs}
\usepackage{pgfplots}          
\pgfplotsset{compat=1.9}

\usepackage{algorithm}
\usepackage{algpseudocode}

\usepackage{verbatim}
\usepackage{enumitem}
\setitemize{noitemsep,topsep=1pt,parsep=1pt,partopsep=1pt}

\newcommand{\system}{{\sc Scrutinizer}\xspace}

\newtheorem{definition}{Definition}

\newtheorem{theorem}{Theorem}
\newtheorem{corollary}{Corollary}

\newtheorem{example}{Example}

 \vldbTitle{Checking statistical claims with a crowd}
 \vldbAuthors{tbd}
 \vldbDOI{https://doi.org/10.14778/xxxxxxx.xxxxxxx}
 \vldbVolume{12}
 \vldbNumber{xxx}
 \vldbYear{2020}

\toappear{ }
\begin{document}

\title{Scrutinizer: A Mixed-Initiative Approach \\to Large-Scale, Data-Driven Claim Verification}

\numberofauthors{1} 
\author{
\alignauthor
Georgios Karagiannis$^\dagger$\thanks{The first two authors contributed equally.} ~~~ Mohammed Saeed$^{\ddagger *}$ ~~~ Paolo Papotti$^\ddagger$ ~~~ Immanuel Trummer$^\dagger$\\
       \affaddr{$^\dagger$Cornell University, USA ~~~~~ $^\ddagger$Eurecom, France}\\
       \affaddr{ \{ gk446,lt224 \}@cornell.edu,  \{ papotti,mohammed.saeed \}@eurecom.fr  }
}

\maketitle

\begin{abstract}
Organizations such as the International Energy Agency (IEA) spend significant amounts of time and money to manually fact check text documents summarizing data. The goal of the Scrutinizer system is to reduce verification overheads by supporting human fact checkers in translating text claims into SQL queries on an associated database. 

Scrutinizer coordinates teams of human fact checkers. It reduces verification time by proposing queries or query fragments to the users. Those proposals are based on claim text classifiers, that gradually improve during the verification of a large document. In addition, Scrutinizer uses tentative execution of query candidates to narrow down the set of alternatives. The verification process is controlled by a cost-based optimizer. It optimizes the interaction with users and prioritizes claim verifications. For the latter, it considers expected verification overheads as well as the expected claim utility as training samples for the classifiers. We evaluate the Scrutinizer system using simulations and a user study, based on actual claims and data and using professional fact checkers employed by IEA. Our experiments consistently demonstrate significant savings in verification time, without reducing result accuracy. 


\end{abstract}

\section{Introduction}
\label{sec:introduction}
Data is often disseminated in the form of text reports, summarizing the most important statistics. For authors of such documents, it is time-consuming and tedious to ensure the correctness of each single claim. Nevertheless, erroneous claims about data are not acceptable in many scenarios as each mistake can have dire consequences. Those consequences reach from embarrassing retractions (in case of scientific papers~\cite{hosseini2018doing}) to legal or financial implications (in case of business or health reports~\cite{ash2004some}). We present \system, a system that helps teams of fact checkers to verify consistency of text and data faster.



Our work is inspired and motivated by the real use case provided by the International Energy Agency (IEA). Every year the agency produces a report of more than 600 pages about the energy consumption and production in the world, covering historical facts and predictions both for individual countries and at the world level. We have been given access to the 2018 edition, which contains 7901 sentences with 1539 manually checked statistical claims. 
Every claim has been checked by three domain experts and their annotations have been collected in a spreadsheet. This process takes months of work of a team of domain experts. Consider the following example from our corpus of statistical claims.



\begin{example}
\label{ex:example1}
The institute has hundreds of relational tables with information about energy, pollution, and climate. A fragment of a table is reported in Figure~\ref{fig:emissions}. Consider the claim ``\textbf{In 2017, global electricity demand grew by 3\%}, more than any other fuel besides solar thermal, reaching 22 200 TWh.". An expert validates the claim in bold by identifying the relevant table(s) and by writing a query over such table to collect the relevant information. In the example:
\newline
\newline
\noindent\fbox{
\parbox{\linewidth}{
SELECT POWER(a.2017/b.2016,1/(2017-2016)) -1

FROM  GED a, GED b

WHERE a.Index = `PGElecDemand', b.Index = 'PGElecDemand'
}}\newline
\newline
Finally, the expert compares the output of the query with the claim and either validates or updates the claim.
\end{example}

\begin{figure}[t]
\begin{center}
\begin{small}
\begin{tabular}{l||l|l|l||l|l|l}
\hline
  Index    & {2017} & 2018 & ...  & {2030} & {2040}\\
     \hline
PGElecDemand    &  22 209    & 22 793  & ...  & 29 349 & 35 526\\
PGINCoal &  2 390    & 2 412 & ..    & 2 341   & 2 353\\
TFCelec     &  21 465  & 22 040 & ...   & 28 566   & 34 790\\
... &  ...  & ... & ...   & ...   & ...\\
\end{tabular}
\end{small}
\end{center}
  \caption{Global Energy Demand history and estimates (GED), the full table has 22 rows and 70 attributes.}
  \label{fig:emissions}
\end{figure}

Gathering data for the claim at hand and composing the right query for the validation takes expertise over the domain and data skills, taking several minutes for a single claim. 
We argue for the need of a system that takes as input the document and a corpus of related datasets to automatically identify the declarative queries that explain why every claim is validated or not by the data. 
\pagebreak
\subsection{Challenges}
Given a document with statistical claims and related datasets, our goal is to come up with the SQL queries that assist the users in the validation, suggesting alternative values for updating a claim in case of an incorrect statement.
Our real-world use case clearly shows three issues that make such data verification hard to automate. 

\vspace{1ex}
\noindent {\bf Text analysis.} Converting a textual claim to a structured query is difficult because claims are expressed in natural language, do not use a fixed vocabulary, and come from multiple authors with different wording and style. 

\vspace{1ex}
\noindent {\bf Query complexity.} Our analysis of the checks done in past by the validation team reveals that the subclass of queries used for checking claim is very wide, going from simple selection to complex mathematical operations involving group of values, aggregations, and functions with more than 100 different combinations of operations.

\vspace{1ex}
\noindent {\bf Large corpus of datasets.} Given a corpus of datasets, it is not clear which one(s) should be used to verify a new statistical claim. In reality, datasets do not come with rich metadata beyond table and attribute names and are heterogeneous in format, schema, and granularity of the data. 


\vspace{1ex}
An exhaustive search of all possible queries is unfeasible, but pruning of the search space must be done carefully. In particular, the testing of false claim is immediately affected, as it is not clear how to judge the inability to create a matching query: is it because of a factual error or from the pruning in the query generation?
With such a difficult problem, we found inspiration from an important resource in our use case. We notice that, by processing the annotations of the checkers, we can collect the data and the operations that have been used to verify every claim. This significant human effort can be used to train models that reduce the search space and identify the queries that verify the claims.

In this direction, we tackle the above challenges with a novel system that builds on three main modules: machine learning (ML) and  natural language processing (NLP) to process the text, human-in-the-loop by involving the domain experts in bootstrapping and validating the candidate queries, and query generation with a large library of functions. The involvement of the users immediately raises more challenges: how to divide the work among a crowd of domain experts? What are the right questions to ask them? How to schedule such questions? How to bootstrap and improve the quality of the models when the training data from previous checks is not available?

\subsection{Contributions}
Translating the claim to a structured query requires to recognize the semantics of the query and the correct data to run it. As the translation is a challenging process and we aim at supporting a large variety of use cases, our system steers the query generation and data matching by generating and scheduling questions to domain experts.

\begin{itemize}[leftmargin=*]
    \item We introduce a novel framework for statistical claims verification that minimizes the human effort (Section~\ref{sec:problem}). \system makes use of classifiers and simple questions to a crowd of domain experts to generate interpretable SQL queries that either validate or contradict the claim (Section~\ref{sec:overview}).   
    \item We build queries by extracting their main features from the textual claim and its context, such as surrounding paragraphs. The classifiers identify the dataset, the attributes, the rows and the mathematical operation that are required to verify the claim. A query generation algorithm combines the provided information into the interpretable queries that are exposed to the user to assess a claim (Section~\ref{sec:querygen}).  
    \item We introduce a cost model and scheduling algorithms for planning the sequence of claims to verify and the questions to ask to crowd of domain experts for a single claim. We give algorithms that minimize the verification cost for the users with quality bounds. The algorithms model the trade-off between the constraints given by the users and the necessity to bootstrap and fine tune the classifiers with labels (Section~\ref{sec:crowd}). 
    \item We experimentally verify with a user study and real data that our system is effective in supporting users in checking claims, enabling the verification of more than {1 claim per minute} on average with a reduction in time of 50\% compared to the original verification process (Section~\ref{sec:exps}). We corroborate those results via simulations, studying performance of different baselines when verifying larger reports.
\end{itemize}



\section{Problem Model}
\label{sec:problem}


We assume a scenario with a crowd of domain experts, a textual document $T$ to be verified, and a set of relational tables $D$.
The textual document is divided into sentences and each sentence $s \in T$ can contain one or more \textit{claims}, that is, word sequences that describe the output of a query $q$ over $D' \subseteq D$. More precisely:

\begin{definition}
A \emph{general claim} describes the comparison $op$ ($<, =, \neq, >$) between the value of query $q$ and a parameter $p$, when $q$ is executed on $D' \subseteq D$. \\
A claim is \emph{correct} if $q(D')\ op\ p$ is true.
\end{definition}

A special instance of our definition of general claim is a common class of statements, where the comparison is the equality and the parameter is a value reported in the claim itself. For equality, we consider a tolerance threshold (admissible error rate) that can be defined by the users.

\begin{definition}
An \emph{explicit claim} 
describes a query $q$ that, when executed on $D' \subseteq D$, returns a value close to the parameter $p$ stated in the claim. An explicit claim is \emph{correct} if the relative difference between $p$ and $q(D')$ is lower than the admissible error rate $e$.
\end{definition}

\begin{example}
\label{ex:example2}
Consider the following two claims: 

\noindent ``\emph{The \textbf{market for new wind power projects increased nine-fold from 2000 to 2017}, \underline{while the solar PV market} \underline{expanded} \underline{aggressively}.
}''

The claim in bold is \emph{explicit} and ``nine-fold'' is the parameter. The query should identify this ratio in the relevant data for wind market (VM) and check an equality, i.e., (VM in 2017 $/$ VM in 2000) = 9. 
The underlined claim is \emph{general}, with ``expanded'' being an operation over solar market (SM) yearly values, i.e., (SM in 2017 $/$ SM in 2000) $>$1, and ``aggressively'' a parameter, i.e., (SM in 2017 $/$ SM in 2000) $>$ 100.
\end{example}

General claims are more challenging than explicit ones because of the ambiguity in the language. The problem is domain specific, as an \textit{aggressive} growth in the energy market may not be the same parameter in the financial or in the automotive market.

Assuming a system can identify the comparison and the parameter in the claim, it still has to come up 
with the correct query. We consider a fragment of SQL focused on statistical checks based on a library of functions $F$ that includes aggregate and mathematical SQL functions, possibly combined with arithmetic operators.

\begin{definition}
\label{def:query}
A statistical check SQL query has the form:\\
{\sc select} $f_i$(a.$A_1$, b.$A_2$, \dots) \\
{\sc from} T1 a, T2 b, $\ldots$ \\
{\sc where} a.key1 = $v_1$ {\sc and} (b.key2 = $v_2$ {\sc or} b.key2 = $v_3$) {\sc and} $\ldots$\\

The {\sc where} clause is a conjunction and disjunction
of $n$ unary equality predicates defined over the key attributes of one or more relations in $D$.
The {\sc select} clause is a (possibly nested) combination of functions $f_i, \ldots, f_m \in F$ defined over attributes values and constants.
\end{definition}

In our use case, we observe that the number of possible function combinations (e.g., {\sc power}(a.2017/b.2016,1/(2017-2016)) in Example~\ref{ex:example1}) is in the hundreds. 
As for the parameter discussion, we do not assume that $F$ is fixed in general, as different combinations are used in different domains.

\begin{example}
Consider again the explicit claim in the previous example about the wind market. The claim is validated if there is a query that translates the textual content and returns a value equals to 9. In this case, the query is
\newline
\newline
\noindent\fbox{
\parbox{\linewidth}{
SELECT (a.2017 / b.2000)

FROM  GED a, GED b

WHERE a.Index = `CapAddTotal\_Wind' {\sc and} b.Index = `CapAddTotal\_Wind';
}}
\end{example}

Finally, we aim at minimizing the effort taken by a group of experts of the domain in verifying the claims in the document. A natural metric to measure the effort is the total time to verify all claims and update incorrect ones. 


We are now ready to formally define our problem
\begin{definition}
Given a set of relations $D$ and a document $T$ containing a set of generic textual claims $C$, we want to minimize the human effort in identifying for every claim $c \in C$ either 
(i) a query $q$ and the relations $D' \subseteq D$  s.t. $c$ is labelled as correct, 
or (ii) that there is no query $q$ and dataset $d$ s.t. $c$ is labelled as correct. 
\end{definition}



In the latter case, we also want to report the 
queries over relations in $D$ that make the claim correct. 

\begin{example}
\label{ex:falseClaim}
Consider again the table fragment in Figure~\ref{fig:emissions} and the (false) claim ``In 2017, global electricity demand grew by {2.5\%}". Our system recognizes that there is no query that returns 2.5\% for global electricity growth in 2017, but there is a query on the same subject returning 3\%. We suggest the value as a possible update to the claim.
\end{example}

\begin{figure}[t]
\centering
\includegraphics[scale=.635]{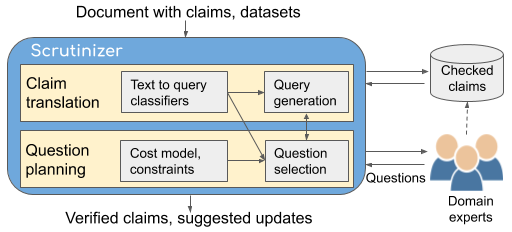}
\vspace{-2.5ex}
\caption{Architecture of \system.}
\label{fig:arch}
\end{figure}

\section{System Overview}
\label{sec:overview}



Figure~\ref{fig:arch} shows an overview of \system. 
The input consists of a text document, containing general claims, and a set of relations. Inspired by our use case, if a database of previously checked claims is available, our system uses it for bootstrapping.
In case such database is not available, we introduce an active learning algorithm to steer the crowd in its creation.
The output of the system is a verification report, mapping verified claims to queries while pointing out mistakes and potential updates to the text. 

The system encompasses two primary components. The automated translation component leverages machine learning to identify the elements that defines every claim, i.e., candidates for datasets, attributes, rows, and comparison operations. 
The question planning component interacts with human domain experts to verify such elements and the checking results, optimizing verification tasks for maximal benefit. 

\begin{algorithm}[ht]
\renewcommand{\algorithmiccomment}[1]{// #1}
\caption{Main verification algorithm.\label{alg:mainAlg}}
\hspace*{\algorithmicindent} \begin{algorithmic}[1]
\State \Comment{Verify claims $C$ in text $T$ using models $M$}
\State \Comment{and return verification results.}
\Function{Verify}{$T,C,M$}
\State \Comment{Initialize verification result}
\State $A\gets\emptyset$
\State \Comment{While unverified claims left}
\While{$C\neq\emptyset$}
\State \Comment{Select next claims to verify}
\State $N\gets$\Call{OptBatch}{$T,C,M$}
\State \Comment{Select optimal question sequence}
\State $S\gets$\Call{OptQuestions}{$N,M$}
\State \Comment{Get answers from fact checkers}
\State $W\gets$\Call{GetAnswers}{$N,M,S$}
\State \Comment{Generate queries and validate claims}
\State $R\gets $\Call{Validate}{$W$}
\State $A\gets W\cup R$
\State \Comment{Remove answered claims}
\State $C\gets C\setminus$\Call{Unanimous}{$N,A$}
\State \Comment{Retrain text classifiers}
\State $M\gets$\Call{Retrain}{$N,A$}
\EndWhile
\State \Comment{Return verification results}
\State \Return{$A$}
\EndFunction
\end{algorithmic}
\end{algorithm}

\begin{figure*}[ht]
\centering
\includegraphics[scale=.325]{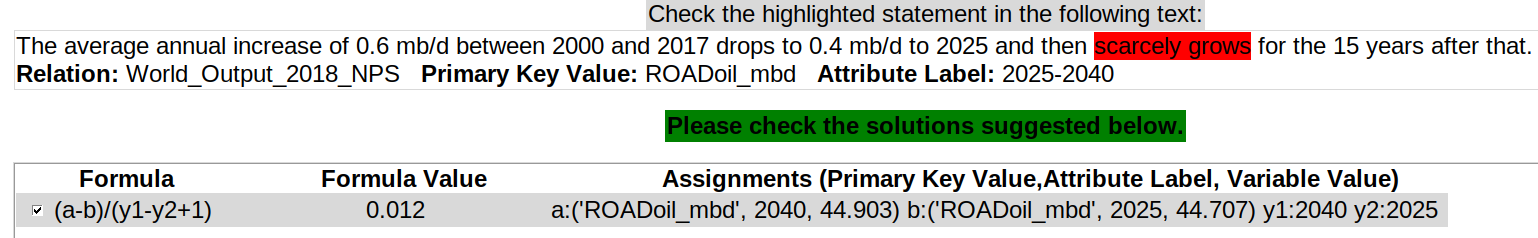}
\vspace{-1ex}
\caption{Example of the generated query (bottom) for the general claim in red in the sentence (top). Below the sentence are reported elements of the claim that have already been validated, such as the database, the key value and the attributes.}
\label{fig:gui}
\vspace{-0.5ex}
\end{figure*}

Algorithm~\ref{alg:mainAlg} describes the main steps in our workflow. Given the claims $C$ in a document $T$ and the ML models, the claims are verified in batches by a team of experts. In each step, the algorithm selects an optimal batch $N$ of claims for verification. Claim batches are selected based on multiple criteria, including expected verification overheads as well as their estimated utility for improving accuracy of the classifiers. For each selected claim in the current batch, we determine an optimal sequence of questions for the human checkers, minimizing expected verification time. Claims are validated or marked as erroneous, based on replies from crowd workers and query evaluations. We remove the claims for which a verification result (i.e., either a verifying query or a decision that the claim is erroneous) can be calculated with sufficiently high confidence. Finally, the classifiers are retrained, based on the newly obtained classification results.


We detail the two main components in the following.

\subsection{Text to Query Translation}
The systems starts by executing four classifiers over the textual claim. We assume the text relevant for the statistical claim has been already identified with one of the existing tools for this task~\cite{JoTY0YLM19}. Given the textual claim, the classifiers identify four elements that are key for the query generation process and claim verification. The first three are basic elements of every query: relevant relations, primary keys values (rows), and attributes names. The fourth classifier is in charge of identifying a generic formula with variables in the place of keys and attribute values. This formula gets instantiated on the dataset at hand and becomes the combination of functions in the SELECT clause. While for explicit claim we always identify the parameter and the comparison, for general claims these two elements can also be predicted within the formula. It is also possible that they cannot be predicted and the user input them answering a question.

\begin{example}
\label{ex:classifiersSimple}
Consider again the (false) claim ``In 2017, global electricity demand grew by \textbf{2.5\%}". Ideally, the first classifier identifies that \emph{global} relations can be used to verify it; the second classifier recognizes that rows reporting values for \emph{electricity demand} should be used; the third classifier returns \emph{2016, 2017} as the attributes of interest, and, finally, the fourth classifier returns the formula {\sc power}$({\frac{a}{b}},{\frac{1}{A_1-A_2}})-1$, with explicit parameter $0.025$ (2.5\%) in the claim (the explicit parameter implies the equality comparison).
\end{example}


To get good accuracy results in the prediction, we resort to active learning. This is in line with our use case, where the previously checked claims are immediately used to derive training data for the classifiers, but also enable the use of our system for cases where previous checks are not available. Previous checks are also important for  generalizing check functions into formulas with variables. This step enable us to (i) reuse formulas on unseen claims and (ii) have a number of classes (for the prediction) as small as possible. 

As we cannot assume that the first prediction is always the correct one in practice, we validate the relations, rows, and attributes predictions with the crowd of domain experts. Once we have this ``context'' information, we predict the top formulas with the last classifier and generate all the possible queries that combine context and formulas. The complexity raised by this combination is in the assignment of the elements of the query to the variables in the formula. Consider two attributes $A_1$ and $A_2$ identified for a certain row and a formula stating that we should compute ``$a-b$'', the system does not know if $A_1$ is assigned to $a$ or $b$. 

\begin{example}
\label{ex:queryGeneration}
Given the predictions for relations (\emph{g1, g2}), rows (\emph{PGElecDemand}), attributes (\emph{2016, 2017}) and formula {\sc power}(${\frac{a}{b}},{\frac{1}{A_1-A_2}})$-1, the query generator module produces all the possible bindings for variables $a,b$ over global relations, for electricity demand rows and with attributes 2016 and 2017. In one assignment, $a$ is bound to a row in relation \emph{g1}, with Index value \emph{PGElecDemand} and attribute \emph{2016}, while in the second assignment $a$ is bound to \emph{g2} and \emph{2016} or \emph{g1} and \emph{2017} and so on. One of these query returns the 3\% parameter in the original claim, thus validating it.
\end{example}


The assignment operation is done in a brute force fashion, but, thanks to the pruning power of the context, it is usually achieved in less than a second.
We describe these components in more detail in Section~\ref{sec:querygen}.

\subsection{Question Planning}
Obtaining feedback from crowd workers is expensive. Hence, the question planning component uses cost-based optimization to determine most effective question sequences. Question planning consists of two sub-tasks. First, for a fixed claim, we choose a sequence of questions allowing us to verify that claim with minimal expected overhead. Each question either solicits crowd workers to verify automatically generated query fragments, or to propose suitable query fragments themselves. Second, we need to decide the order in which claims are verified. When selecting claims to verify next, we take into account expected verification overheads as well as their value as training samples for our classifiers (used for automated claim verification). We describe this component in more detail in Section~\ref{sec:crowd}.

\begin{example}
\label{ex:gui}
Figure~\ref{fig:gui} shows at the bottom the query generated after a group of relations, a key value and two attributes have been validated for the general claim at the top. The domain expert can examine the formula that has been predicted (left), its assignment over all the relations that contain the key value and attributes (right), and the resulting value (0.012 in the example) for verifying the claim. 
\end{example}

Notice that in the example above the parameter is not predicted by the formula, the user has to assess if 0.012 is correctly described by ``scarcely".

We remark that our system is designed for a setting with many claims that need to be verified by a team of checkers. If this is not the case, it does not add an extra cost but the effort in training the classifiers would not be visible.

\section{Claim Translation}
\label{sec:querygen}

We first describe how we preprocess the claims to extract the features to be used with the classifiers. We then describe 
the query generation step.


\subsection{Claim Preprocessing}
Given a text, we start by processing it to identify worth checking claims with existing tools~\cite{HassanALT17,JaradatGBMN18}.
Given a claim (sequence of words) it is necessary to identify the correct relations from the corpus. In such relations, we need to identify the primary key values and the attributes that identify the data values to be used in the check operation. For these three tasks, we rely on (\textit{GloVe}) pre-trained embeddings to convert the text to a distributed representation which maps each word to a real-valued vector~\cite{pennington2014glove}. 
To get the embedding of a sentence, we average the embedding of each word in that sentence. For each claim in a sentence, we concatenate the sentence embedding with the TF-IDF scores of the unigrams and bigrams in the claim, followed by the TF-IDF scores of every 3 characters. 

\begin{figure}[th]
\vspace{-0.5ex}
\centering
\includegraphics[scale=.55]{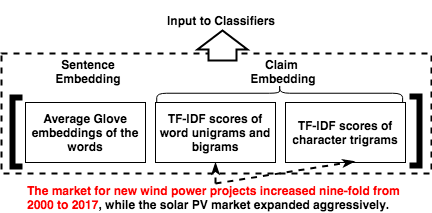}
\vspace{-2ex}
\caption{Preprocessing of the claims.}
\label{fig:embeddings}
\vspace{-0.5ex}
\end{figure}

As depicted in Figure~\ref{fig:embeddings}, embeddings for the sentence and the claim are fed as multi-dimensional vectors to the four classifiers responsible for predicting a fragment of the final query. One classifier is used to predict a list of possible relations that are used to verify the claim. Another classifier predicts a list of primary key values that are relevant to the claim. A third classifier predicts a list of possible attribute labels. The final classifier predicts a list of possible operations. If the claim is explicit, 
we identify the parameter $p$ directly from the sentence with a syntactical parsing. 






\subsection{From Claims to Formulas}
Given the large variety of possible statistical checks, we do not rely on a pre-defined library of operations and their possible combinations, but learn them from the previously checked claims.
Given a previously checked claim, we describe how we turn it into a generic formula with variables, with the goal of reusing it with unseen claims. A classifier is trained with (claim, formula) pairs and returns a ranked list of formulas for a given textual claim.

\begin{example}
\label{ex:formula}
A query with\\ 
SELECT POWER(a.2017/b.2016,1/(2017-2016))-1\\
identifies a formula
POWER($a$.$A_1$/$b$.$A_2$,1/($A_1$-$A_2$))-1
\end{example}

Example~\ref{ex:formula} shows the translation of the Select clause of a query into a formula. The formula contains variables for the relations and for the attributes, but preserves function names, operations, and constants. The variables make the check reusable for a new claim, assuming that the classifier returns the correct formula for it. In a formula, one of the operations and one of the constants can play the role of comparison $op$ and parameter $p$, respectively, for claims that are not explicit.

Given annotations with SQL queries, the process to obtain a formula is straightforward.
Unfortunately, going from previous checks to formulas is challenging because we do not assume SQL queries in the annotations. In fact, in our use case, checkers used spreadsheets and notes in natural language to annotate their verification process. The lack of rigorous guidelines raises three problems.

\vspace{1ex}
\noindent \textbf{Reconstruction.}
We call \textit{look-up} the function that retrieves \textit{data values} from the relations (Select and Project in SQL). Given a relation, each data value is identified by its primary key value (e.g., ``PGElecDemand'') and its attribute name (e.g., ``2017'' or ``Total''). Data values in a claim can be collected from different relations. 
Check operations range from simple look-ups in the relations to compositions of SQL functions and operations, possibly with constants. 
Moreover, values can be results of other intermediate operations. Any value involved in an operation may be obtained from operations such as subtraction, multiplication, or even compound annual growth rate\footnote{\url{https://en.wikipedia.org/wiki/Compound\_annual\_growth\_rate}} that involve looking-up other values. 

Any value might be the result of several operations, 
thus formulas contain the entire sequence of operations. We achieve this by recursively replacing each value by its corresponding function in the annotations until we reach a look-up. As attribute labels are present in some formulas as values, we also replace them with attribute variables.

\vspace{1ex}
\noindent \textbf{Ambiguity.}
A complication from the lack of guidelines is that checkers verify the same claims with different operations. Even for simple explicit claims, one checker may write a Boolean query and see if the output is empty, while another may collect a value from the data and compare it visually with the parameter. The problem gets harder with general claims, as in the following example.

\begin{example}
\label{ex:ambiguity}
Consider an explicit claim stating that the consumption of some resource $r$ in a certain year $y$ has been ``very high". One checker may verify if with a Boolean query:
\newline
\newline
\noindent\fbox{
\parbox{\linewidth}{
SELECT d.$y > 100$
FROM  rel
WHERE d.key=$r$ 
}}
\newline
but a second checker may verify the claim with query:\\
\newline
\noindent\fbox{
\parbox{\linewidth}{
SELECT d.$y$ FROM  rel WHERE d.key=$r$
}}\newline
and marking the claim as correct based on some parameter that is neither in the claim nor in the query.
\end{example}

\vspace{1ex}
\noindent \textbf{Incomplete information.}
The second check in Example~\ref{ex:ambiguity} shows a case of incomplete annotation for a general claim. This is quite common in our experience, as also shown in Example~\ref{ex:example1} and
in the formula reported in Figure~\ref{fig:gui}: only a human can conclude that 0.012 is an appropriate value for this domain and claim to validate ``scarcely''.
While for explicit claims is not an issue, as the comparison and the parameter are in the claim, incomplete annotations lead make it impossible to even replicate a past check on the same general claim. 
This problem clearly motivates our human-in-the-loop solution, detailed in Section~\ref{sec:crowd}.

\vspace{1ex}
Due to the first two problems above, it is hard to generate formulas in practice, as it is reflected by our experimental results with the prediction of formulas from the text. Fortunately, we use information from the claim data to better identify the correct formula, as we discuss next.

\begin{algorithm}[t]
\renewcommand{\algorithmiccomment}[1]{// #1}
\caption{Query generation algorithm.}
\label{alg:query_gen_algo}
\hspace*{\algorithmicindent} \begin{algorithmic}[1]
\State \Comment{Given relations $R$, key $K$, attributes $A$, formulas $F$,}
\State \Comment{parameter $p$ (for explicit claim), returns queries}
\Function{GenerateQueries}{$R$,$K$,$A$,$F$,$p$}
\State $V, S, S_A \gets\emptyset$
\State \Comment{Collect data value assignments $V$}
\For {$r \in R$, $k \in K$, $a \in A$}
\State $V \gets V \cup$ \Call{GetValue}{$r,k,a$}
\EndFor
\For{$f \in F$}
\State \Comment{Get \# of non-attribute variables in formula}
\State  $n \gets$\Call{GetVars}{$f$}  
\State  $P \gets$\Call{GetPerm}{$V,n$}  \For{$i \in P$}
\State \Comment{Test approx. value match for explicit claim}
\If{$p \neq \emptyset$ \textbf{and} $f(i) \approx p$}
\State $S \gets S \cup i$ 
\ElsIf{$S \neq \emptyset$} 
\State $S_A \gets S_A \cup i$
\EndIf
\EndFor
\EndFor
\State \Comment{Rewrite variables and assignments as queries}
\If{ $S \neq \emptyset$}  
\State $Q \gets $ \Call{Rewrite}{$R,K,A,S$}
\State \Call{Return}{$Q$}
\Else
\State $Q_A \gets$ \Call{Rewrite}{$R,K,A,S_A$}
\State  \Call{Return}{$Q_A$}
\EndIf
\EndFunction
\end{algorithmic}
\end{algorithm}


\subsection{Query Generation}
We describe the query generation process in Algorithm~\ref{alg:query_gen_algo}. The input of the algorithm is the output of the classifiers and the parameter $p$ if the claim is explicit.
To generate the candidate queries, we first get all the possible values from the combinations of the classifier outputs for relations, primary key values, and attribute labels (line \textbf{7}). Then, we loop through the list of formulas (line \textbf{9}), and for each formula, we get the number of possible permutations (line \textbf{12}) of the possible values. We then try each permutation (line \textbf{13}) to see if it leads to a match for the explicit claim 
and eventually store it as a solution (line \textbf{16}). If we did not find a valid solution or the claims was not explicit, then we store the solution in a different list (line \textbf{18}). After looping through the formulas, if a solution was found for the explicit claim, we produce the queries associated to these solutions (line \textbf{24}). In all other cases, 
we produce queries for all solutions (line \textbf{27}). 
In the rewriting, we fill up a query template with the relations, key values, attribute labels and formula instantiated. The query template is an SQL string with placeholders, as described in Definition~\ref{def:query}. Note that we generate \textit{(SELECT-PROJECT-AGGREGATE)} queries 
that can span multiple relations.
Finally, we return all queries (lines \textbf{25}, \textbf{28}).



The algorithm assumes that the input information for relations, key values and attributes are correct as these come from the crowd validation, as described in the next section. Formulas are not validated by the crowd as returned by the classifier, but only after they have been filtered in the instantiation loop in the algorithm.



\begin{example}
Consider the following input to the algorithm:

\vspace{0.3ex}
\noindent
\textbf{Relations}: GED; 
\textbf{Keys}: PGElecDemand,              
\textbf{Attributes}: 2016, 2017;
\textbf{Formulas}: Power$(\frac{a}{b}, {\frac{1}{A_1-A_2}})-1$, $(a+b)>0$, \ldots 



After instantiating the first formula and 
replacing the query template, we obtain:
\newline
\newline
\noindent\fbox{
\parbox{\linewidth}{
SELECT POWER(a.2017/b.2016,1/(2017-2016)) -1

FROM  GED a, GED b

WHERE a.Index = `PGElecDemand', b.Index = `PGElecDemand'
}}
\end{example}

In the last step of the workflow, the queries are executed and the results displayed to the user to draw conclusions on the claim, as depicted in Figure~\ref{fig:gui}.

\section{Question Planning}
\label{sec:crowd}
Question planning consists of two tasks: determining optimal questions to verify single claims, and determining an optimal verification order between claims. We discuss the first problem in Section~\ref{sub:verification} and the second one in Section~\ref{sub:ordering}.

\subsection{Single Claim Verification}
\label{sub:verification}
We verify claims by asking a series of questions to human fact checkers. Our goal is to minimize overheads for the fact checkers. To do so, we leverage the results of our claim to query translation components. In the ideal case, we have identified a query that translates the current claim with high confidence. In that case, crowd workers only need to verify the proposed translation. This is typically faster than verifying the claim manually.

In practice, we are not always able to find a high-confidence translation for a claim. Instead, we may still be able to narrow down the range of possibilities to a small set of alternatives. If this is not possible for the query as a whole, we may still be able to do so for specific query properties (e.g., we identify specific columns that appear in the query with high confidence). In those cases, we can ask crowd workers to verify assumptions about specific query properties, or to select answers from a small set of options. Of course, answering questions on query properties or selecting answers causes overheads as well. Our goal is to select the sequence of questions that minimize expected verification cost.

For each claim, we generate a series of screens. Each screen contains questions that are answered by a crowd worker. Each screen is associated with one specific query property (e.g., the presence of specific columns or tables). On the upper part of each screen, crowd workers are shown a set of answer options with regards to the current property. Those answer options are obtained from our classifiers. On the lower part of each screen, crowd workers have the option to suggest new options, if the correct answer is not on display. The final screen for each claim asks directly for the query translating the current claim. Answers to prior questions may have allowed us to narrow down the range of possible queries. If so, the chances for confronting workers with the correct query increase.

In this scenario, our search space for question planning is the following. First, we need to decide how many screens to show. Second, we need to determine what query properties our questions should focus on. Third, we need to decide how many answer options to display on each screen. Fourth, we need to pick those answer options. 

We make those decisions based on a simple cost model, representing time overhead for crowd workers for verifying the current claim. We assume that workers read screen content from top to bottom. For each answer option, a worker needs to determine whether it is correct or not. We count a per-option verification cost in our model, distinguishing cost of verifying answers about query properties, $v_p$, from the cost of verifying the full query (on the final screen), $v_f$. We choose constants such that $v_p\ll v_f$ to account for the fact that full queries are significantly longer than their fragments (which increases reading time and therefore verification cost). If none of the given options applies, crowd workers must suggest an answer themselves. We denote by $s_p$ and $s_f$ the cost of suggesting answers for properties and queries (again, $s_p\ll s_f$). 

First, we discuss how to choose the number of screens and answer options. We denote the number of screens by $nsc$ and the number of options by $nop$. Predicting the precise verification cost for specific choices of those parameters is not possible. Doing so would require knowing the right solution to each question (as it determines how many options workers will read). However, we can upper-bound verification cost in relation to the cost of verifying claims without \system.

\begin{theorem}
Compared to the baseline, relative verification overhead of \system is at most $(nop\cdot v_f+nsc\cdot(v_p+s_p))/s_f$.\label{th:claimcost}
\end{theorem}
\begin{proof}
Reading through answer options on the final screen adds cost overheads of $nop\cdot v_f$ in the worst case. We have overheads of $nsc\cdot (v_p+s_p)$ for all previous screens. Verifying the claim without help means suggesting a query for the current claim. This has cost $s_f$ in our model. 
\end{proof}

\begin{corollary}
Setting $nop=s_f/v_f$ and $nsc=s_f/(v_p+s_p)$ limits verification overheads to factor three.
\end{corollary}
\begin{proof}
This follows immediately by substituting the proposed formulas in the equations from Theorem~\ref{th:claimcost}.
\end{proof}

We will use the aforementioned setting for most of our experiments. Having determined the number of screens and options, we still need to pick specific screens and answers. First, we discuss the selection of answer options. Note that the worst-case verification cost of a property depends only on the number of options shown (but not on the options themselves). Hence, to pick options, we consider expected verification cost instead. 

We calculate expected verification cost based on our classifiers, assigning specific answer options to a probability. For a fixed property, denote by $A$ the set of all relevant answer options. Also, denote by $p_a$ the probability that an answer $a\in A$ is correct. We calculate expected verification cost when presenting users with an (ordered) list of answer options $\langle a_1,\ldots,a_m\rangle$ where $a_i\in A$. 

\begin{theorem}
The expected verification cost for answer options $\langle a_1,\ldots,a_m\rangle$ is $v_p\cdot\sum_{i=1..m}(1-\sum_{1\leq j< i}p_{a_i})$.\label{th:verificationExp}
\end{theorem}
\begin{proof}
We consider the case that at most one answer option is accurate (this case is typical). The cost of verifying one answer option is $v_p$ (assuming properties). The probability that workers need to read beyond the $i$-th option is the probability that none of the first $i$ options is correct: $\Pr(\text{$a_1$ to $a_i$ incorrect})=1-\sum_{1\leq j\leq i}p_{a_i}$. The expected cost is the cost of each verification, weighted by the probability that it is necessary: $v_p\cdot\sum_{i=1..m}(1-\sum_{1\leq j< i}p_{a_i})$.
\end{proof}

\begin{corollary}
Selecting answer options in decreasing order of probability minimizes expected verification cost.
\end{corollary}
\begin{proof}
Each term in the cost formula, proven in Theorem~\ref{th:verificationExp}, decreases if the sum of probabilities of the first options increases. Hence, starting with higher probability choices decreases cost.
\end{proof}

Finally, we discuss the selection of query properties. Our goal is to select the best $nsc$ properties to verify by creating corresponding screens. 

We define the quality of a property as follows. At any point (before verification), we consider a set of possible query translations for a claim. A large set of possible query translations is problematic for two reasons. First, it leads to higher computational overheads when executing them to obtain tentative result. Second, we increase overheads for fact checkers who may be presented with a large number of alternatives. A good property has high pruning power with regards to the current set of candidates. This means that it allows us to discard as many incorrect candidates as possible. 

How many query candidates we can prune depends on the actual property value. Depending on the answer we obtain from the fact checkers, more or less queries can be pruned. We do, of course, not know the correct answers when selecting questions. Hence, we define the expected pruning power of a set of properties as follows. 

\begin{definition}
Given a set $Q$ of query candidates, a set $S$ of query properties to verify, and trained models $M$ predicting a-priori probabilities for possible answers, we define the pruning power $\mathcal{P}(S,Q,M)$ as the expected number of queries that are excluded by obtaining answers for $S$. 
\end{definition}

Next, we provide a formula for pruning power, based on simplifying assumptions. For that, we denote by $a_s^i$ the $i$-th answer option for property $s\in S$ and by $E_s^i\subseteq Q$ queries that are excluded if answer option $a_s^i$ turns out to be correct.

\begin{theorem}\label{th:pruningPower}
The pruning power $\mathcal{P}(S,Q,M)$ is given by $\sum_{q\in Q}(1-\prod_{s\in S}\sum_{i:q\notin E_s^i}\Pr(\text{$a_s^i$ correct}|M))$.
\end{theorem}
\begin{proof}
The pruning power is given as the expected number of pruned queries: $\sum_{q\in Q}\Pr(\text{$q$ is pruned})$. Clearly, it is $\Pr(\text{$q$  pruned})=1-\Pr(\text{$q$ not pruned})$. We simplify by assuming independence between properties and obtain \[\Pr(\text{$q$ not pruned})=\prod_{s\in S}\Pr(\text{$q$ not pruned by $s$}|M)\]. Furthermore, we assume that different answer options for the same property are mutually exclusive. Then, we obtain \[\Pr(\text{$q$ not pruned by $s$})=\sum_{i:q\notin E_s^i}\Pr(\text{$a_s^i$ correct}|M)\]. Substitution yields the postulated formula.
\end{proof}

Next, we discuss the question of how to find property sets maximizing above formula. Iterating over all possible property sets is possible but expensive (exponential complexity in the number of properties). Instead, we select properties according to a simple, greedy approach. At each step, we add whichever property maximizes pruning power to the set of selected properties (when comparing properties to add, we calculate pruning power for the union between the new and previously selected properties). We stop once the number of selected properties has reached the threshold determined before. While this algorithm may seem simple, it offers surprisingly strong formal guarantees. Those guarantees are derived from the fact that pruning power is a sub-modular function~\cite{Nemhauser1978}. We define sub-modularity below. 

\begin{definition}
A set function $f:S\mapsto\mathbb{R}$ is sub-modular if, using $\Delta_f(S,s)=f(S\cup\{s\})-f(S)$, it is $\Delta_f(S_1,s)\geq\Delta_f(S_2,s)$ for any $S_1\subseteq S_2$.
\end{definition}

Intuitively, sub-modularity captures a ``diminishing returns'' behavior. If adding more elements to a set, the utility of new elements decreases as the set of previous elements grows. The pruning power function is sub-modular as well, according to the following theorem.

\begin{theorem}
Pruning power is sub-modular.\label{th:pruningSubmodular}
\end{theorem}
\begin{proof}
Consider the probability that one specific query is not pruned via questions relating to any property, given as $\prod_{s\in S}\Pr(\text{$q$ not pruned by $s$}|M)$ (see proof of Theorem~\ref{th:pruningPower}). From the perspective of each query, adding one more property corresponds to multiplying its probability of not being pruned by a factor between zero and one. For $x_1,x_2,y\in[0,1]$, it is generally $x_1-x_1\cdot y\geq x_2-x_2\cdot y$ if $x_1\geq x_2$. As the probability of not being pruned does not increase when adding questions, the impact of adding a new question on pruning probability decreases for each query. This means the probability of one query of being pruned is sub-modular in the question set. The same applies to pruning power itself (as a sum over sub-modular functions with positive weights is sub-modular).
\end{proof}

Next, we show that the simple greedy algorithm produces a near-optimal set of questions.

\begin{theorem}
Using the greedy algorithm, we select a set of questions that achieve pruning power within factor $1-1/e$ of the optimum.
\end{theorem}
\begin{proof}
The greedy algorithm is equivalent to the greedy algorithm by Nemhauser~\cite{Nemhauser1978}. The pruning power function is sub-modular (see Theorem~\ref{th:pruningSubmodular}), it is non-negative (as we sum over probabilities) and non-decreasing (as pruning probability can only increase when adding more questions). Hence, it satisfies the conditions under which those bounds have been proven for Nemhauser's algorithm~\cite{Nemhauser1978}.
\end{proof}

Finally, we analyze time complexity (denoting by $nsc$ the number of screens, by $npr$ the number of properties, and by $nqu$ the number of query candidates).

\begin{theorem}
Finding optimal question sequences for verifying single claims is in $O(nsc\cdot npr\cdot nqu)$. 
\end{theorem}
\begin{proof}
The greedy algorithm performs $O(nsc)$ steps and considers $O(npr)$ options in each step. Evaluating the pruning power function requires $O(nqu)$ steps if using a naive approach (we can reduce complexity if query candidates are represented by a Cartesian product between query properties). 
\end{proof}

The complexity of selecting optimal question sequences for claims is therefore polynomial in all problem dimensions. This is important, as we need to re-run this step for each claim in the document, whenever classifiers are retrained. This is due to the fact that expected verification cost, based on the optimal question sequence, forms the input to the algorithm discussed next.

\subsection{Claim Ordering}
\label{sub:ordering}

Next, we discuss the problem of determining a claim order for verification. At first, it may not be clear why verification order matters. If modeling verification cost per claim as a constant, total verification cost is simply the cost sum over all claims. In that model, verification order does not matter indeed.

However, verification cost per claim is not static. As time progresses, the quality of automated claim translation increases (as claims verified by crowd workers serve as training samples). This decreases expected claim verification cost at the same time (as crowd workers merely need to assert proposed claim translations). Hence, verifying claims in different order may indeed influence overall verification cost. 

We consider two criteria when selecting the next claims to verify. First, we consider the benefit of claim labels for training our classifiers (for automated claim to query translation). Second, we consider the expected verification cost.

The first point relates to prior work on active learning. Here, the goal is generally to select optimal training samples to increase the quality of a learned model. In our case, verified claims correspond to training samples for classifiers that translate claims to queries. Picking training samples with maximal uncertainty (according to the current model) is a popular heuristic in the context of active learning. We follow this approach as well and define the training utility as follows. 

\begin{definition}
Let $m\in M$ a model predicting specific properties of the query associated with a text claim $c$. We assume that $m$ maps each claim to a probability distribution over property values. Denote by $e(m,c)$ the entropy of that probability distribution. We define the training utility of $c$, $u(c)$ by averaging over all models (associated with different query properties): $u(c)=\sum_{m\in M}e(m,c)$.
\end{definition}

The second point (verification cost) relates to the cost model discussed in the previous subsection. However, this cost model is incomplete. It neglects the cost of understanding the context in which a certain claim is placed. Intuitively, verifying multiple claims in the same section is faster than verifying claims that are far apart in the input document. Our extended cost model takes this into account. In contrast to the model from the previous subsection, it calculates verification cost for claim batches (instead of single claims).

\begin{definition}
Denote by $C$ a batch of claims for verification. For each claim $c\in C$, denote by $s(c)$ the section in which this claim is located (instead of sections, a different granularity such as paragraphs can be chosen as well). Denote by $v(c)$ the pure claim verification cost for $c$ defined in the last subsection. Further, denote by $r(s)$ the cost of reading (respectively skimming) section $s$. We define the total (combined verification and skimming) cost for claim batch $C$ as the sum of both verification cost over all claims and reading cost over all associated sections: $t(C)=(\sum_{c\in C}v(c))+(\sum_{s\in \{s(c)|c\in C\}}r(s))$.
\end{definition}

This cost model has the desired property: it captures the fact that verifying claims in the same section is faster. Our approach to claim ordering is based on this model. It is not useful to determine a global claim order before verification starts. We cannot predict how the quality of classifiers (and therefore claim verification cost) will change over time. Instead, we repeatedly select claim batches that are presented to the checkers. Those claim batches are selected based on training utility and the aforementioned cost model. 

Note that we prefer selecting claim batches as opposed to single claims. First, presenting fact checkers with claim batches allows them to better plan their verification strategy. For instance, claims can be clustered in a first pass to treat claims that are semantically related together during verification. Second, integrating new training samples and optimally selecting claim batches is computationally expensive. As discussed next, selecting claim batches is a hard optimization problem. Also, retraining all classifiers (the operation that motivates re-running claim selection) is a relatively expensive operation on our test platform. By re-training on claim batches, rather than single claims, we reduce computational overheads.

To select claim batches, we solve the following optimization problem.

\begin{definition}
Given a set of unverified claims $C$, the goal of claim selection is to select a claim batch $B\subseteq C$ such that total cost of $B$ remains below a threshold $t_m$: $t(B)\leq t_m$. Additionally, the minimal and maximal batch size is restricted by parameters $b_l$ and $b_u$: $b_l\leq |B|\leq b_u$. Under those constraints, the goal is to maximize accumulated training utility $\sum_{c\in B}u(c)$. Alternatively, as a variant, we minimize the cost formula $t(B)-w_u\cdot \sum_{c\in B}u(c)$ where $w_u$ is a weight representing the relative importance of selecting claims with high uncertainty for classifier training. 
\end{definition}

This problem is hard, as shown by the following theorem.

\begin{theorem}
Claim selection is NP-hard.
\end{theorem}
\begin{proof}
We prove NP-hardness by a reduction from the knapsack problem. Let $I=\{\langle w_i,b_i\}$ a set of items with associated weights $w_i$ and benefit $b_i$. The goal is to maximize accumulated benefit $\sum_{i\in I^*}b_i$ for an item set $I^*\subseteq I$ whose accumulated weight remains below a threshold $T$: $\sum_{i\in I^*}b_i\leq T$. We construct an equivalent instance of claim selection as follows. We introduce an unverified claim $c_i$ for each item $i\in I$. We assume that each claim is located in a separate section ($s_i$ for claim $c_i$). We set combined verification and reading cost for each claim and associated section to be proportional to item weight: $v(c_i)+r(s_i)=w_i$. Training utility is proportional to benefit ($u(c_i)=b_i$). We choose cardinality bounds that do not influence the solution ($b_l=0$ and $b_u=|I|$). Now, an optimal solution to claim verification yields an optimal solution to the original knapsack instance (via a polynomial time transformation).
\end{proof}

The fact that claim selection is NP-hard justifies the use of sophisticated solver tools. We reduce the problem to integer linear programming. This allows us to apply mature solvers for this standard problem. Next, we discuss how we transform claim selection into integer linear programming.

An integer linear program (ILP) is generally characterized by a set of integer variables, a set of linear constraints, and a (linear) objective function. The goal is to find an assignment from variables to values that minimizes the objective function, while satisfying all constraints. 

We introduce binary decision variables of the form $cs_i$, indicating whether the $i$-th claim was selected ($cs_i=1$) or not ($cs_i=0$). Also, we introduce binary variables of the form $sr_j$ to indicate whether section number $j$ needs to be skimmed or not (to verify the selected claims). Next, we express the constraints of our scenario on those variables. First, we limit the number of selected claims to the range $[b_l,b_u]$ by introducing the linear constraints $b_l\leq \sum_i cs_i\leq b_u$. Next, we represent the constraint that sections of selected claims must be read. We introduce constraints of the form $sr_j\geq cs_i$ if claim $i$ is located within section $j$. Furthermore, we limit accumulated verification cost of the selected claims by the constraint $(\sum_i cs_i\cdot v(c_i))+(\sum_j sr_j\cdot r(s_j))$. Finally, we set $-\sum_i cs_i\cdot u(c_i)$ as objective function to minimize. 

The time complexity for solving a linear program generally depends on the solver (and the algorithm it selects to solve a
specific instance). However, the number of variables and constraints often correlates with solution time. We analyze both in the following.

\begin{theorem}
The size of the ILP problem is in $O(cc\cdot sc)$ where $cc$ is the claim count and $sc$ the section count.
\end{theorem}
\begin{proof}
The number of variables is in $O(cc+sc)$ while the number of constraints (specifically: constraints connecting claims to sections read) is in $O(cc\cdot sc)$.
\end{proof}

The ILP size grows relatively slowly in the number of claims and sections. While claim selection remains NP-hard, we show in our experiments that we can solve corresponding problem instances sufficiently fast in practice.

\section{Experiments}
\label{sec:exps}
We evaluated \system using real data along two dimensions: (i) the effectiveness of the tool in real verification tasks with domain experts, (ii) the effectiveness
and efficiency of question scheduling. The code of the system is available online\footnote{\url{https://github.com/geokaragiannis/statchecker}}.

\vspace{1ex}
\noindent {\bf Dataset.} We obtained a document of 661 pages, containing 7901 sentences, and the corresponding corpus of manually checked claims, with check annotations for every claim from three domain experts.
The annotations cover 
1539 numerical claims, of which about half are explicit. The massive effort in checking claims is because the document authors write the report with early estimates, so the data underlying the book change over time. In the first pass of the draft, up to 40\% of the claims are updated.

\begin{table}[ht]
    \centering
    \caption{Percentiles of property value frequencies.}
    \begin{tabular}{llllll}
    \toprule[1pt]
    \textbf{Percentiles} & \textbf{10\%} & \textbf{25\%} & \textbf{50\%} & \textbf{95\%} &\textbf{99\%}\\
    \midrule[1pt]
    \textbf{Relation} & 2 & 4 & 10 & 199 & 532\\
    \midrule
     \textbf{Primary Key} & 2 & 2 & 4 &39 &107 \\
    \midrule
     \textbf{Attribute} & 1 & 2 & 7 & 127 & 1400\\
     \midrule
     \textbf{Formula} & 1 & 1 & 1 & 8 & 55\\
    \bottomrule[1pt]
    \end{tabular}
    \label{tab:percentiles}
\end{table}


After processing the claims, we identify 1791 relations, 830 key values, 87 attribute labels, and 413 formulas. Table~\ref{tab:percentiles} shows some percentiles of the frequency distribution of each property. We see that 50\% of the values for all properties appear at most 10 times in the corpus, with the top 5\% most frequent formulas appearing at least 8 times.

\subsection{User Study}
\label{sub:user_study}

In this experiment we involved seven domain experts from the institution to measure the benefit of our system compared to the traditional manual workflow for verification. 
We trained \system with all the annotated statistical claims and randomly selected 43 claims among the ones with the 10 formulas that cover the majority of the claims. As we only have access to the correct version of the claim, we randomly selected 25\% of them to inject errors.

Three experts have been randomly assigned to the \textbf{Manual} process and the remaining four to the \textbf{System}-assisted process. We gave them instructions to execute the test without interruptions and without collaboration. Three claims (two correct, one incorrect) have been used for training on the new process and the remaining 40 for the study. The task given to the experts was to verify as many claims as possible in 20 minutes, given access to their traditional tools in the manual process (spreadsheets and databases) and to our system only in the second case. The order of the claims has been fixed to allow comparison among experts and the time for checking every claims has been registered.

\begin{figure}[t]
\begin{tikzpicture}
\begin{axis}[
    height=5cm,
    width=8cm,
    ymajorgrids,
	ylabel=\# Claims,
	xlabel= Checkers,
	legend style={at={(0.5,-0.28)},anchor=north,legend columns=-1},
	enlargelimits=0.1,
 	ybar stacked,
 	xtick=data,
        /pgfplots/nodes near coords*/.append style={
            every node near coord/.style={
                name=X,
                shift={(9.5pt,8pt)}
            },
            scatter/@post marker code/.append code={
                \node(Y){};
            }
        },
	nodes near coords,
	every node near coord/.append style={font=\footnotesize},
	bar width=10pt,
    symbolic x coords={M1,M2,M3,S1,S2,S3,S4}
]
	\addplot coordinates
		{(M1,10) (M2,13) (M3,8) (S1,19) (S2,26) (S3,23) (S4,20)};
	\addplot coordinates
		{(M1,0) (M2,0) (M3,0) (S1,1) (S2,2) (S3,2) (S4,1)};
	\addplot coordinates
		{(M1,0) (M2,0) (M3,2) (S1,1) (S2,0) (S3,0) (S4,1)};
\legend{Correct,Incorrect,Skipped}
\end{axis}
\end{tikzpicture}
\caption{Number of claims verified in 20 minutes by checkers with the manual process (M1--M3) and with our system (S1--S4).}
\label{fig:users}
\end{figure}
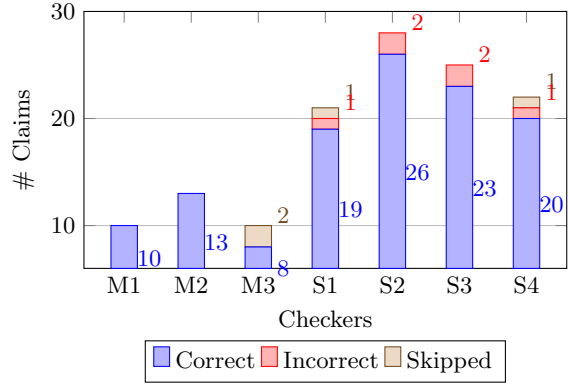

We distinguish three cases: skipped claims, claims that have been correctly labelled, and incorrect decisions. Results for each checker are reported in Figure~\ref{fig:users}. Considering correct and incorrect checks, on average a user verifies 7 claims manually and 23 claims with \system in 20 minutes. Users tend to skip a comparable amount of claims in both settings. In the System process, a few claims have been incorrectly checked. Those are all correct claims labelled as incorrect. However, with a simple majority voting across any subset of three checkers, our system obtain 100\% accuracy as in the manual process. 
There was only one claim where verification time using the tool surpassed the traditional manual verification time. After investigating this special case, it turned out that this was due to sequential checking. The user consulted a relation different from what we were expecting him to choose. The different relation led him to the correct answer, but such relation was used also in the previous question with the same primary key and attributes values, making this claim very fast to verify.

By using simple majority voting over three checkers, the accuracy of the aggregate answers is 100\%  for both the Manual and the System groups.

\begin{figure}
\begin{tikzpicture}
\begin{axis}[
    height=5cm,
    width=8.5cm,
    ymajorgrids,
	xlabel= Claim complexity,
	ylabel= Verification time (secs),
	legend style={at={(0.712,0.91)},
	anchor=north,legend columns=-1}]
\addplot[line width=1.64pt, color=blue, error bars/.cd, y dir=both, y explicit] coordinates {(3,98)+=(3,61)-=(3,61) (6,105)+=(6,80)-=(6,80)};
\addplot[line width=1.64pt, color=red, error bars/.cd, y dir=both, y explicit] coordinates {(3,44)+=(3,14)-=(3,14) (6,50)+=(6,19)-=(6,19) (8,60)+=(8,19)-=(8,19) 
(10,50)+=(10,13)-=(10,13) (11,78)+=(11,57)-=(11,57)};
\legend{Manual, System}
\end{axis}
\end{tikzpicture}
\caption{Average time to verify claims of increasing complexity with the Manual and System processes.}
\label{fig:time_complexity}
\end{figure}
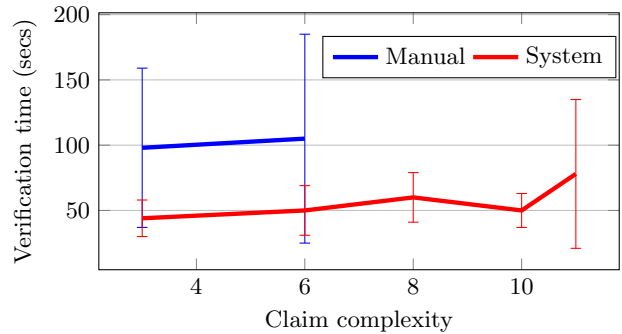

We also report in Figure~\ref{fig:time_complexity} average verification time and standard deviation for the two groups of checkers with claims of increasing complexity.
The claim complexity is the sum of the elements in the query to verify it: number of key values, attributes, operations, constants and variables. Checkers using \system take on average less than half the time to verify claims of the same complexity. The average time taken by checkers using our system for claims of 11 elements is lower that the time used by checked with the manual process with 6 elements. We report in the plot claims for which at least two checkers have been able to process it. We are therefore not showing in the plot a checker using \system who took on average 29 seconds to verify two claims of complexity 14. We remark that for one claim of complexity 6, it took 203 for one of the Manual users to verify it, while for the same claim the slower System user spent 66 for the same task. 

We conducted the study on a laptop (1.80GHz x 8 i7 CPU, 32 GB of memory).
For any claim, testing a classifier took less than 0.2 seconds and query generation took less than half a second (0.35 seconds on average).

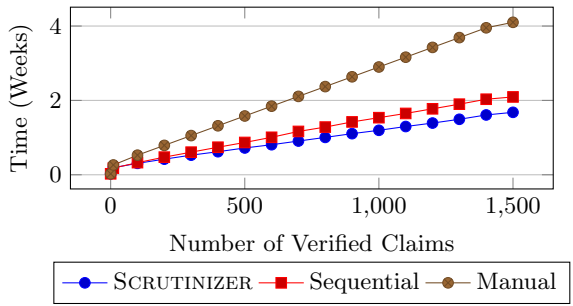
\begin{figure}
\begin{tikzpicture}
\begin{axis}[
    height=4cm,
    width=8cm,
    ymajorgrids,
	xlabel= Number of Verified Claims,
	ylabel= Time (Weeks),
	legend style={at={(0.5,-0.4)},
	anchor=north,legend columns=-1}]
\addplot 
	table[x=num_training_data,y=milp_sum_cost_weeks,col sep=comma]{files/plots_base_seq_milp_costs+acc.csv};
\addplot 
	table[x=num_training_data,y=seq_sum_cost_weeks,col sep=comma]{files/plots_base_seq_milp_costs+acc.csv};
\addplot 
	table[x=num_training_data,y=sum_baseline_cost_weeks,col sep=comma]{files/plots_base_seq_milp_costs+acc.csv};
\legend{\system, Sequential, Manual}
\end{axis}
\end{tikzpicture}
\caption{Accumulated verification time over verification period.}
\label{fig:cost_experiment}
\end{figure}

\subsection{Simulation}
In the previous subsection, we have demonstrated that \system decreases verification overheads for single claim batches. Next, we study the efficiency of \system when verifying entire reports. Verification time for entire reports is typically in the order of months for IEA. Hence, we cannot use another user study. Instead, we created a simulator, based on the results of our initial user study. We simulate the verification of the 2018 IEA world energy outlook report, using the original claims and original data. We assume a team of three fact checkers (which is typical for IEA). We simulate a ``cold start'' scenario, meaning that our classifiers have no initial training data. Instead, they use claim labels provided by simulated fact checkers. This corresponds to the worst case for our system. It represents a scenario in which the very first version of a new report is received and verified. Our model for verification time per claim is based on time measured in the user study. It takes into account reduced verification overheads once proposed query fragments are accurate.
We compare three baselines. First, we consider manual verification (``Manual'') which is the current default. Each claim is verified without any computational support. Second, we consider a simplified version of \system. This version (``Sequential'') does not optimally reorder claims, as described in Section~\ref{sub:ordering}, but verifies them sequentially (i.e., in document order) instead. We compare those two approaches against the \system system. For Sequential and \system, we assume that ten answer options are shown per property. For \system, we use claim batches of size 100 (after which we retrain classifiers and select the next claims to verify via ILP). Our simulator is implemented in Python~3, using Gurobi~9.0.1 as ILP solver. Experiments were executed on a MacBook Pro with 2.4~GHz Intel Core i5 processor and 8~GB memory.

\begin{table}[ht]
    \centering
    \caption{Summary of simulation results.}
    
    \begin{tabular}{llll}
    \toprule[1pt]
    \textbf{} & \textbf{Manual} & \textbf{Sequential} & \textbf{Scrut.}\\
    \midrule[1pt]
    \textbf{Time (Weeks)} & 4.1 & 2.1 & 1.7\\
    \midrule
     \textbf{\% Savings} & - & 49\% & 59\%\\
    \midrule
     \textbf{Avg.  Accuracy} & - & 40\% & 47\%\\
     \midrule
     \textbf{Max  Accuracy} & - & 46\% & 53\%\\
     \midrule
     \textbf{Comp.\ (Mins)} & - & 14 & 28\\
    \bottomrule[1pt]
    \end{tabular}
    \label{tab:simulation_summary}
\end{table}

\begin{figure}[t]
\begin{tikzpicture}
\begin{axis}[
    height=4cm,
    width=8cm,
    ymajorgrids,    
	xlabel= Number of Verified Claims,
	ylabel= Accuracy (\%),
	legend style={at={(0.5,-0.4)},
	anchor=north,legend columns=-1}]
\addplot 
	table[x=num_training_data,y=milp_acc,col sep=comma]{files/plots_base_seq_milp_costs+acc.csv};
\addplot 
	table[x=num_training_data,y=seq_acc,col sep=comma]{files/plots_base_seq_milp_costs+acc.csv};
\legend{\system, Sequential}
\end{axis}
\end{tikzpicture}
\caption{Evolution of \system and sequential average accuracy over verification period.}
\label{fig:simulation_acc}
\end{figure}
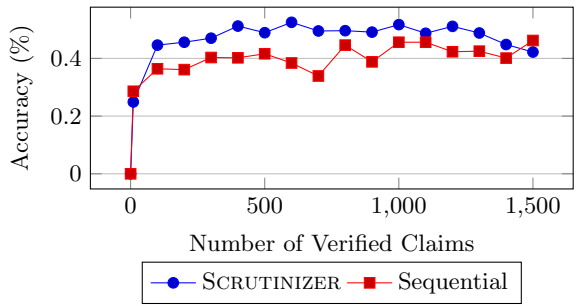

Table~\ref{tab:simulation_summary} summarizes simulation results. We report total verification time for all three fact checkers, assuming an eight hours work day and a five day week. We make the following observations. First, using \system reduces verification overheads by more than factor two (circa 60\%). This is consistent with the results of our user study. At the same time, it is remarkable since we consider a cold start scenario. The results show that, given a sufficiently large document to verify, the initial warmup period of the classifiers does not impact overall performance by too much. Second, we observe a positive impact due to claim ordering. While using \system without that feature is still helpful, cost savings increase when claims are systematically prioritized. Table~\ref{tab:simulation_summary} shows that, in the latter case, the average (and maximal) quality of classification over the entire period improves. Figure~\ref{fig:cost_experiment} shows that \system and the sequential baseline are near-equivalent at the beginning of the classification process. Claim ordering pays off more and more as verification proceeds. At the same time, computational overheads are negligible for all compared systems. \system spends 15 minutes in total to plan optimal question sequences, and for selecting optimal claims via ILP. The remaining 13 minutes are due to retraining classifiers.

Figure~\ref{fig:simulation_acc} digs deeper and analyzes classifier accuracy, as a function of verification time. We compare \system and the sequential baseline. The accuracy of \system dominates the one of the sequential version over most of the verification period. The only exceptions appear at the very beginning and towards the end of the verification process. This can be explained as follows. \system makes an upfront investment by selecting the claims for which classifiers are most uncertain. This translates into higher per-claim verification costs since proposed query fragments are often incorrect. On the other side, classifiers learn faster which increases accuracy for the following batches. Once classification accuracy increases, the term evaluating claims according to their utility as training samples in the ILP objective function decreases. Instead, terms related to expected verification costs become dominant. \system then selects claims that have lower expected verification costs. It postpones handling of difficult claims, associated with low classifier confidence, towards the very end. Once only those claims are left, accuracy of \system drops below the one of the sequential version. The effect of this drop on verification cost is however negligible.

Figure~\ref{fig:milp_acc} decomposes accuracy for \system (with claim ordering) according to classifier type. The effect discussed in the last paragraph (a steep increase followed by a drop towards the end) still hold when considering classifiers separately. Further, we notice that certain properties are harder to infer from text. For instance, inferring row indices is among the hardest classification tasks. This can be explained by the fact that the classification domain (i.e., number of rows) is typically larger than for other classifiers (e.g., columns). Also, different data subsets often have a similar structure (translating into the same set of columns) which is not necessarily the case for the row keys.

Finally, in Figure~\ref{fig:topk_acc}, we analyze accuracy for the top-k labels and for different classifiers. In most cases, classifiers reach most of their potential with the first 10 entries. 

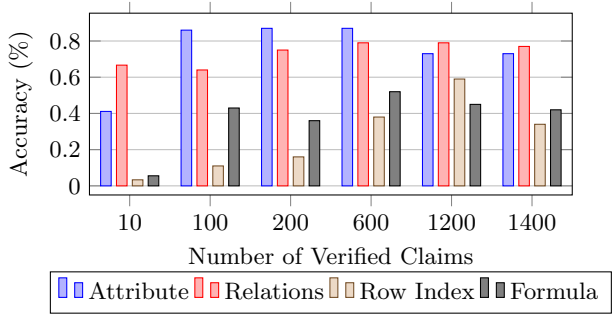
\begin{figure}[t]
\begin{tikzpicture}
\begin{axis}[
    height=4cm,
    width=8cm,
    ymajorgrids,
	ylabel=Accuracy (\%),
	xlabel=Number of Verified Claims,
	legend style={at={(0.5,-0.42)},
		anchor=north,legend columns=-1},
	enlargelimits=0.1,
	ybar,
	bar width=4pt,
    symbolic x coords={10,100,200,600, 1200, 1400}
]
\addplot
    table[x=num_training_data,y=column,col sep=comma]{files/milp_individual_classifiers.csv};
\addplot
    table[x=num_training_data,y=file,col sep=comma]{files/milp_individual_classifiers.csv};
\addplot
    table[x=num_training_data,y=row_index,col sep=comma]{files/milp_individual_classifiers.csv};
\addplot
    table[x=num_training_data,y=template_formula,col sep=comma]{files/milp_individual_classifiers.csv};

\legend{Attribute,Relations,Row Index,Formula}
\end{axis}
\end{tikzpicture}
\caption{Evolution of classifier accuracy over verification period.}
\label{fig:milp_acc}
\end{figure}

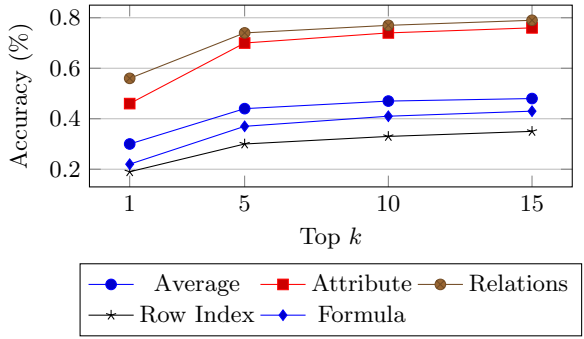
\begin{figure}[t]
\begin{tikzpicture}
\begin{axis}[
    xtick={1,5,10,15},
    height=4cm,
    width=8cm,
    ymajorgrids,
    xlabel=Top $k$,
    ylabel=Accuracy (\%),
    legend style={at={(0.5,-0.42)},
		anchor=north,legend columns=3}]
\addplot
    coordinates {
	(1, 0.3)
	(5, 0.44)
	(10, 0.47)
	(15, 0.48)
};
\addplot
    coordinates {
	(1, 0.46)
	(5, 0.70)
	(10, 0.74)
	(15, 0.76)
};
\addplot
    coordinates {
	(1, 0.56)
	(5, 0.74)
	(10, 0.77)
	(15, 0.79)
};
\addplot
    coordinates {
	(1, 0.19)
	(5, 0.30)
	(10, 0.33)
	(15, 0.35)
};
\addplot
    coordinates {
	(1, 0.22)
	(5, 0.37)
	(10, 0.41)
	(15, 0.43)
};
\legend{Average, Attribute, Relations, Row Index, Formula}
\end{axis}
\end{tikzpicture}
\caption{Top $k$ accuracy for different \system classifiers as a function of $k$. }
\label{fig:topk_acc}
\end{figure}

\section{Related Work}
\label{sec:related}
The automatic verification --aka fact checking-- of textual claims is a wide research topic, with different proposals covering different aspects. A first division is between two tasks: (1) spotting worth-checking claims~\cite{HassanALT17,JaradatGBMN18}, (2) verifying claims.
\system targets the second task.

More precisely, our solution uses external information in the form of reference data. This is different from approaches that exploit Web documents~\cite{MihaylovaNMBMKG18,ThorneV18,Wang0BK18} or 
databases of previously checked claims~\cite{KaragiannisTJKW19,HassanZACJGHJKN17},
but also from the study of misinformation spread in social networks~\cite{Sherchan:2013:STS:2501654.2501661,FerraraVDMF16}.

We also distinguish the verification of property claims and statistical claims.
Simple property claims, such as ``Berlin is the capital of Germany'' have been studied in the context of reference information expressed in (incomplete) knowledge bases (KBs)~\cite{ciampaglia2015computational,shi2016discriminative,gardner2014incorporating,gardner2015efficient,DBLP:journals/pvldb/HuynhP19}. Here the main challenge comes from a claim that is validated by facts in the KB at hand. The check consists in deriving from existing information if the fact is missing due to incompleteness or because it is incorrect. An important aspect is the creation of a human-consumable explanations for the fact checking decisions~\cite{Leblay17,Gad-Elrab0UW19,expclaim2019}; \system goes in this direction as declarative queries are easy to parse for users.

\begin{table}[ht]
    \centering
    \caption{Properties of the systems. 
}
    \begin{tabular}{lllll}
    \toprule[1pt]
& {\sc Scruti-} & \textsc{Agg} & \textsc{BriQ} & \textsc{Stat}\\
& {\sc nizer} & \textsc{Check.}\cite{JoTY0YLM19} & \cite{IbrahimRWZ19} & \textsc{Search}\cite{CaoMT18}\\
    \midrule[1pt]
\textbf{Task} & check & check  & check   & search \\
 & $n$ claims  & 1 claim & 1 claim  & 1 claim\\
    \midrule
\textbf{Claims} & general & explicit & explicit   & explicit\\
    \midrule
\textbf{Query} & SPA + & SPA + & SPA + &  SP \\
 & 100s ops   & 9 ops  &    6 ops & \\
\midrule
\textbf{User} & crowd  & single & single   & single\\
\midrule
\textbf{Dataset} & corpus  & single & single   & corpus\\
    \bottomrule[1pt]
    \end{tabular}
    \label{tab:related}
\end{table}


\system relates to prior efforts on data-driven fact checking of statistical claims~\cite{CaoMT18,IbrahimRWZ19,JoTY0YLM19} with differences in the main properties as reported in Table~\ref{tab:related}. The closest work is the AggChecker system~\cite{JoTY0YLM19} in that it translates statistical claims into SQL queries for verification. It differs however by the following characteristics. First, \system supports a richer query model than prior work (including implicit queries) and is able to learn new query templates during the verification process. This is necessary as complex queries with long arithmetic expressions are omnipresent in the IEA data. Second, \system is targeted at the verification of large documents with hundreds of pages. Such documents are typically verified by teams over extended periods of time. This focus motivates the design as a crowd system~\cite{LiWZF17,FranklinKKRX11}, as well as optimizations such as active learning, that only pay off in the long term. Also, \system relates to prior work on mixed-initiative fact checking~\cite{Gatterbauer2009}; however, they do not consider claims that are verified by executing queries on structured data. It is complementary to prior work on verifying robustness of claims that are initially given as SQL queries~\cite{wu2014toward} (as, in our scenario, the main challenge is the translation from text to queries). Finally, our work relates to prior work in the domain of natural language query interfaces~\cite{AgrawalCD02,LiJ14,Zhong2017}. As pointed out in prior work~\cite{JoTY0YLM19}, verification creates new challenges and opportunities. For instance, we can exploit values that appear in claim text to narrow down query candidates.

\section{Conclusion}
\label{sec:conclusion}
We introduced \system, the first system for crowdsourcing the verification of general statistical claims. Our solution effectively minimize the amount of work needed by a group of domain experts to verify a large variety of textual claims in a document. Experimental results show that our algorithms enable the cold start of the system and leads to the automatic generation of queries that are easy to parse and validate by humans. Given the results of our user study, where \system reduces the verification time to less than 50\% compared to the current workflow, the IEA is planning to start using our system for their next edition of the yearly report, in the summer of 2020.



\bibliographystyle{abbrv}
\bibliography{FC}

\balance

\end{document}